\documentclass[11pt]{article}
\usepackage{amssymb,amsmath}

\usepackage{fullpage}

\usepackage{times}
\usepackage{helvet}
\usepackage{courier}
\usepackage{amssymb,amsmath}
\usepackage{comment}
\usepackage{graphicx}

\usepackage{mathabx}

\usepackage{amsmath}
\usepackage{comment}
\usepackage{amssymb}

\usepackage{authblk}

\usepackage{tikz}
\usetikzlibrary{arrows}

\newcommand{\ignore}[1]{}

\begin{document}

\title{The parameterized space complexity of embedding along a path}

\author[1]{Hubie Chen}
\author[2]{Moritz M\"{u}ller}
\affil[1]{
University of the Basque Country (UPV/EHU), 
E-20018 San Sebasti\'{a}n, Spain,
\emph{and} IKERBASQUE, Basque Foundation for Science, E-48011 Bilbao, Spain\\
  \texttt{hubie.chen@ehu.es}}
\affil[2]{ Kurt G\"{o}del Research Center, University of Vienna, Austria \\
  \texttt{moritz.mueller@univie.ac.at}}
\date{ }

\maketitle

\newcommand{\cl }{\mathcal }

\begin{abstract}
The embedding problem is to decide, given an ordered pair of structures,
whether or not there is an injective homomorphism from the 
first structure to the second.  
We study this problem using an established perspective in parameterized
complexity theory: the universe size of the first structure is taken
to be the parameter, and we define the embedding problem relative
to a class ${\cl A}$ of structures to be the restricted version
of the general problem where the first structure must come
from ${\cl A}$.  
We initiate a systematic complexity study of this problem family,
by considering   classes whose structures are 
what we call rooted path structures;
these structures 
have paths as Gaifman graphs.  
Our main theorem is a dichotomy theorem on classes of
rooted path structures.
 \end{abstract}


\newtheorem{theorem}{Theorem}[section]
\newtheorem{conjecture}[theorem]{Conjecture}
\newtheorem{corollary}[theorem]{Corollary}
\newtheorem{proposition}[theorem]{Proposition}
\newtheorem{prop}[theorem]{Proposition}
\newtheorem{lemma}[theorem]{Lemma}
\newtheorem{remark}[theorem]{Remark}
\newtheorem{exercisecore}[theorem]{Exercise}
\newtheorem{examplecore}[theorem]{Example}
\newtheorem{examples}[theorem]{Examples}

\newenvironment{example}
  {\begin{examplecore}\rm}
  {\hfill $\Box$\end{examplecore}}

\newenvironment{exercise}
  {\begin{exercisecore}\rm}
  {\hfill $\Box$\end{exercisecore}}

\newenvironment{proof}{\noindent\textbf{Proof\/}.}{$\Box$ \vspace{1mm}}

\newtheorem{researchq}{Research Question}

\newtheorem{newremarkcore}[theorem]{Remark}

\newenvironment{newremark}
  {\begin{newremarkcore}\rm}
  {\end{newremarkcore}}

\newtheorem{definitioncore}[theorem]{Definition}

\newenvironment{definition}
  {\begin{definitioncore}\rm}
  {\end{definitioncore}}

\newcommand{\str }{\mathbf }
\newcommand{\ar }{\mathit{ar} }
\newcommand{\atyp }{\mathrm{atyp} }
\newcommand{\pl}{\mathrm{pl}}
\newcommand{\refl}{\mathrm{refl}}

\newcommand{\emb}[1]{\ensuremath \textsc{emb}(#1)}
\newcommand{\pemb}[1]{\ensuremath p\textsc{-emb}(#1)}
\newcommand{\phom}[1]{\ensuremath p\textsc{-hom}(#1)}

\newcommand{\lgsh}{\textsc{longshort}}

\newcommand{\paraL}{\textup{para-L}}
\newcommand{\PATH}{\textup{PATH}}
\newcommand{\TREE}{\textup{TREE}}

\newcommand{\ustcon}{\textsc{ustcon}}
\newcommand{\lastcon}{\textsc{lastcon}}

\newcommand{\npprob}[4]{
\begin{center}
\begin{tabular}{|r p{10cm} |}
\hline
\multicolumn{2}{|l|}{\textsc{#1}}\\
\textit{Instance:}& #2.\\
\textit{Parameter:}& #3.\\
\textit{Problem:}& #4\\
\hline
\end{tabular}
\end{center}
}

\maketitle

\section{Introduction}

The \emph{embedding problem} is to decide,
given a pair $(\str A, \str B)$ of structures,
whether or not there is an 
\emph{embedding}---an injective homomorphism---from $\str A$
to $\str B$.
Intuitively, the embedding problem asks whether one can identify
a particular type of pattern, specified by the first structure $\str A$,
in the second structure $\str B$. 
Of course, this is a problem of a fundamental nature.
Indeed, a number of well-established computational problems can be
viewed as cases of this problem.
%
%
Examples include the problems {\sc clique, cycle} and {\sc path} which 
ask, given an undirected graph $\str G$ and a natural number $k$, whether $\str G$ 
contains a size $k$ clique, length $k$ cycle or length $k$ path, respectively.
%
%

In parameterized complexity theory,
the embedding problem is typically studied 
by taking the universe size of the first structure $\str A$
as the parameter;
this is the perspective and parameterization that we use here.
In the examples above, this corresponds to 
the standard parameterization
by the value $k$, and 
yields the famous 
parameterized problems 
$p$-{\sc clique}, $p$-{\sc cycle} and $p$-{\sc path}.
While this parameterized problem is in general intractable 
($p$-{\sc clique} is W[1]-complete),
it has been fruitful
to consider the following family of restricted versions of the problem:
for each class ${\cl A}$ of structures,
$\pemb{\cl A}$ is the parameterized embedding problem where the first 
structure $\str A$ must come from ${\cl A}$.
Throughout, we assume that each  class of structures 
under discussion is on a 
shared finite vocabulary.
%
In a now-famous result~\cite{AlonYusterZwick95-colorcoding}, 
it was established that
the problem $\pemb{\cl A}$ is fixed-parameter tractable
when ${\cl A}$ has bounded treewidth.
The algorithmic technique introduced there, called \emph{color coding},
in fact can be viewed as providing a Turing reduction from
the embedding problem $\pemb{\cl A}$ to 
the homomorphism problem $\phom{\cl A^*}$.
Here, 
the problem $\phom{\cdot}$ is defined analogously to $\pemb{\cdot}$,
but asks merely for a homomorphism 
(as opposed to an injective homomorphism);
 the class ${\cl A^*}$ is obtained from ${\cl A}$
by replacing each structure ${\str A}$ in ${\cl A}$
with the structure ${\str A^*}$,
 which is the structure ${\str A}$ but expanded so that there is, for each element $a$ of ${\str A}$, a unary relation symbol $U_a$ interpreted as~$\{ a \}$.

The family $\phom{\cl A}$ of homomorphism problems is well-understood.
A classification of these problems 
up to parameterized logarithmic space reduction is known~\cite{ChenMueller15-fineclass},
which shows that each problem $\phom{\cl A}$ is
either in para-L (parameterized logarithmic space), 
PATH-complete, TREE-complete, or W[1]-complete.
The complexity classes PATH and TREE 
are subclasses of FPT, and indeed the inclusions
para-L $\subseteq$ PATH $\subseteq$ TREE $\subseteq$ FPT
are known
(further discussion of these classes 
can be found in~\cite{ChenMueller15-fineclass}).
 It has been shown that para-L and PATH are not equal under the assumption that 
 Savitch's classical simulation cannot be improved~\cite{ChenMueller14-paramlogspace}.

%
%
%
%

In contrast, the family $\pemb{\cl A}$ of embedding problems
seems quite enigmatic.  While it has been conjectured
that
 the problem $\pemb{\cl A}$ is W[1]-hard
 when ${\cl A}$ does not have bounded treewidth
(see for example~\cite[p.355]{FlumGrohe06-parameterizedcomplexity}),
this research issue is (to our knowledge) wide open.
Indeed, only recently was the complexity of the
prominent problem {\sc biclique} 
resolved as 
W[1]-complete~\cite{Lin15-biclique};
this can be defined as the particular problem $\pemb{\cl A}$
where ${\cl A}$ is the class of complete bipartite graphs.
Let us mention that, concerning our examples, 
it is known that
$p$-{\sc path} is in $\paraL$ \cite{ChenMueller14-paramlogspace} and
that
$p$-{\sc cycle} is $\PATH$-complete \cite{ChenMueller15-fineclass}.

\subparagraph*{Contributions}
The motivation behind the present work was to initiate a 
systematic study of the $\pemb{\cl A}$ family of problems,
in hopes of eventually obtaining classification results
of the form known for the homomorphism problem.
We here focus on classes $\cl A$ of 
\emph{rooted path structures}. These are structures whose Gaifman graph is a path and in addition are {\em rooted} in the sense that one of its endpoints is the sole element of a relation.
While this implies that 
$\pemb{\cl A}$ is in PATH (and hence in FPT), 
the suggestion here is to first obtain
a thorough understanding of the problem family 
with respect to small complexity classes, and then 
attempt to scale up this understanding.

As examples of our findings, consider the following 
three classes of structures.

\begin{enumerate}
\item The class of rooted alternating paths.

\begin{tikzpicture}

\tikzset{vertex/.style = {shape=circle,draw}}
\tikzset{edge/.style = {->,> = latex'}}
\node[vertex,fill] (a) at  (0,0) {};
\node[vertex] (b) at  (1.5,0) {};
\node[vertex] (c) at  (3,0) {};
\node[vertex] (d) at  (4.5,0) {};
\node[vertex] (e) at  (6,0) {};

\node[vertex] (f) at   (7.5,0) {};
\node[vertex] (g) at   (9,0) {};
\node[vertex] (h) at   (10.5,0) {};
\draw[edge] (a) to (b);
\draw[edge] (c) to (b);
\draw[edge] (c) to (d);
\draw[edge] (e) to (d);

\draw[edge] (f) to (g);
\draw[edge] (h) to (g);

\node[draw=none] at (6.75,0) {$\cdots$};
\end{tikzpicture}

\item The class derived from rooted alternating paths by subdividing
each edge.

\begin{tikzpicture}

\tikzset{vertex/.style = {shape=circle,draw}}
\tikzset{edge/.style = {->,> = latex'}}
\node[vertex,fill] (a) at  (0,0) {};
\node[vertex] (ap) at  (0.75,0) {};

\node[vertex] (b) at  (1.5,0) {};
\node[vertex] (bp) at  (2.25,0) {};

\node[vertex] (c) at  (3,0) {};
\node[vertex] (cp) at  (3.75,0) {};

\node[vertex] (d) at  (4.5,0) {};
\node[vertex] (dp) at  (5.25,0) {};

\node[vertex] (e) at  (6,0) {};

\node[vertex] (f) at   (7.5,0) {};
\node[vertex] (fp) at   (8.25,0) {};
\node[vertex] (g) at   (9,0) {};
\node[vertex] (gp) at   (9.75,0) {};
\node[vertex] (h) at   (10.5,0) {};
\draw[edge] (a) to (ap);
\draw[edge] (ap) to (b);
\draw[edge] (c) to (bp);
\draw[edge] (bp) to (b);
\draw[edge] (c) to (cp);
\draw[edge] (cp) to (d);
\draw[edge] (e) to (dp);
\draw[edge] (dp) to (d);

\draw[edge] (f) to (fp);
\draw[edge] (fp) to (g);
\draw[edge] (h) to (gp);
\draw[edge] (gp) to (g);

\node[draw=none] at (6.75,0) {$\cdots$};
\end{tikzpicture}

\item The class derived from rooted alternating paths by adding one final non-alternating edge: 

\begin{tikzpicture}

\tikzset{vertex/.style = {shape=circle,draw}}
\tikzset{edge/.style = {->,> = latex'}}
\node[vertex,fill] (a) at  (0,0) {};
\node[vertex] (b) at  (1.5,0) {};
\node[vertex] (c) at  (3,0) {};
\node[vertex] (d) at  (4.5,0) {};
\node[vertex] (e) at  (6,0) {};

\node[vertex] (f) at   (7.5,0) {};
\node[vertex] (g) at   (9,0) {};
\node[vertex] (h) at   (10.5,0) {};
\node[vertex] (i) at   (12,0) {};
\draw[edge] (a) to (b);
\draw[edge] (c) to (b);
\draw[edge] (c) to (d);
\draw[edge] (e) to (d);

\draw[edge] (f) to (g);
\draw[edge] (h) to (g);
\draw[edge] (i) to (h);

\node[draw=none] at (6.75,0) {$\cdots$};
\end{tikzpicture}
\end{enumerate}

It follows from our results that, with respect to the embedding problem,
 the first class is in para-L, whereas the second and third one are PATH-complete.
The complexity of these classes can be derived from
a dichotomy theorem 
(Theorem~\ref{thm:dichotomy-rooted-oriented-paths})
that characterizes 
classes of rooted path structures where each
structure is an oriented path with a root.

Our main theorem is a dichotomy theorem 
which describes the complexity of the problem $\pemb{\cl A}$
for each class of rooted path structures 
(Theorem~\ref{thm:main}).
However, we do not succeed in obtaining a para-L versus PATH-complete
dichotomy, as in the previously mentioned dichotomy theorem.
Instead, for each such problem $\pemb{\cl A}$,
we either show it to be
$\PATH$-complete or we exhibit a parameterized logarithmic space algorithm that solves the problem with oracle access to a problem which we call the \emph{long-short path problem}.
This algorithmic result  is based on color coding and Reingold's algorithm.

A number of remarks are in order.
First, if one shows that the long-short path problem
is in para-L, our positive complexity result can immediately
be improved to containment in para-L.
Second, with respect to the problem family considered,
the long-short path problem is both unavoidable and occurs naturally
in the family,
in the following precise sense: there exists a class of rooted path
structures ${\cl A}$ such that $\pemb{\cl A}$ is
equivalent, under parameterized logarithmic space Turing reduction,
to the long-short path problem 
(Theorem~\ref{thm:long-short-as-embedding-problem}).
Hence, one necessarily needs to resolve the complexity
of the long-short path problem in order to describe all problems
in the studied problem family up to parameterized logarithmic space
Turing reduction.
Third, independently of what the complexity of the long-short path problem
turns out to be, we believe that the present work
makes a contribution in identifying and isolating the
long-short path problem as the
hardest of the embedding problems 
in our family that are of unknown complexity.
This identification can indeed be conceived of as a 
form of completeness result. 
Our view is that settling the complexity of this concrete problem
is a challenge to known techniques, and thus that a deeper understanding
of this problem
could mark healthy progress in the understanding of parameterized
logarithmic space.



To get some feeling for the difficulty, we encourage the reader to ponder whether the embedding problem associated with the following class of rooted path structures is in para-L or PATH-complete.  The class consists in undirected paths prolonged by alternating paths. In a picture:

\begin{enumerate}
\item[4.] 
\begin{tikzpicture}

\tikzset{vertex/.style = {shape=circle,draw}}
\tikzset{edge/.style = {->,> = latex'}}
\node[vertex,fill] (a) at  (0,0) {};
\node[vertex] (b) at  (1.5,0) {};
\node[vertex] (c) at  (3,0) {};
\node[vertex] (d) at  (4.5,0) {};
\node[vertex] (e) at  (6,0) {};
\node[vertex] (f) at   (7.5,0) {};
\node[vertex] (g) at   (9,0) {};
\node[vertex] (h) at   (10.5,0) {};
\node[vertex] (i) at   (12,0) {};
\node[vertex] (j) at   (13.5,0) {};

\draw[edge] (a) to (b);
\draw[edge] (b) to (a);
\draw[edge] (c) to (b);
\draw[edge] (b) to (c);
\draw[edge] (c) to (d);
\draw[edge] (d) to (c);
\draw[edge] (e) to (f);
\draw[edge] (f) to (e);

\draw[edge] (f) to (g);
\draw[edge] (h) to (g);
\draw[edge] (h) to (i);
\draw[edge] (j) to (i);

\node[draw=none] at (14.25,0) {$\cdots$};
\node[draw=none] at (5.25,0) {$\cdots$};
\end{tikzpicture}

\end{enumerate}
This gives the maybe the simplest (and most annoying) example of an embedding problem for rooted path structures which 
we conjecture to belong to para-L but are only able to  reduce  to the long-short path problem.

To close this introduction, let us make the following observations.
As mentioned, it is known that
when the class ${\cl A}$ has
bounded treewidth, 
 the problem
$\pemb{\cl A}$ is in FPT;
it has been conjectured that the problem $\pemb{\cl A}$
is W[1]-hard otherwise.
This conjecture thus suggests that 
one need only look 
at the
Gaifman graphs of the structures in a class ${\cl A}$
to determine whether or not $\pemb{\cl A}$ is in FPT.
Under the assumption that the complexity degrees dealt with
in this article are pairwise distinct,
our results contrast sharply with this suggestion:
we only consider structures with path
Gaifman graphs, but show that within the realm of such structures,
dichotomies occur, 
and hence the Gaifman graph does not carry the information needed
to determine the complexity of $\pemb{\cl A}$.
Moreover, the aforementioned conjecture implies that
if $\pemb{\cl A}$ is in FPT at all, then it is in FPT
via using color coding to reduce to $\phom{\cl A^*}$.
Our positive complexity results go strictly beyond this paradigm
of reducing to $\phom{\cl A^*}$ 
because $\phom{\cl A^*}$ is always $\PATH$-complete for any infinite class~$\cl A$ of rooted path structures over the same finite vocabulary
(this follows from~\cite{ChenMueller15-fineclass}).
%

\section{Preliminaries}

\subsection{Structures and logic}
A {\em (relational) vocabulary} is a finite set $\tau$ of relation symbols; every $R\in\tau$ has an associated 
{\em arity} $r\in\mathbb N$. Recall {\em $\tau$-formulas} are built from atomic $\tau$-formulas by means of $\wedge,\neg$ and $\exists x$, and  
an {\em atomic} $\tau$-formula has the form $R(x_1,\ldots x_{r})$ or $x_1=x_2$ where the $x_i$ are variables and $R$ is an $r$-ary relation 
symbol from $\tau$.
The notation $\varphi(x_1,\ldots, x_r)$ means that the free variables of the $\tau$-formula $\varphi$ are among $x_1,\ldots, x_r$. 

A {\em $\tau$-structure} $\str A$
consists of a  non-empty {\em universe} $A$ and  for every $r$-ary
relation symbol $R\in\tau$ a relation $R^{\str A}\subseteq A^{r}$.
We only consider structures with finite universes. 
A {\em (induced) substructure} of $\str A$ is a $\tau$-structure $\str B$ with $B\subseteq A$
and $R^{\str B}= R^{\str A}\cap B^r$ for every $r$-ary $R\in\tau$.

If $\varphi(x_1,\ldots, x_r)$ is a $\tau$-formula, $\str A\models\varphi(a_1,\ldots, a_r)$ means that the tuple $(a_1,\ldots, a_r)\in A^r$ 
satisfies $\varphi(x_1,\ldots, x_r)$ in~$\str A$. 
The {\em atomic type $\atyp(\bar a,\str A)$}  of a tuple $\bar a=(a_1,\ldots, a_r)\in A^r$ is the set of atomic 
$\tau$-formulas $\varphi(x_1,\ldots, x_r)$ such 
that $\str A\models\varphi(a_1,\ldots, a_r)$. If $\str A$ is clear from context, we write $\atyp(\bar a)$ instead $\atyp(\bar a,\str A)$.

We view graphs as $\{E\}$-structures $\str G$ for $E$ a binary relation symbol such that $E^{\str G}$ is irreflexive and symmetric.
Elements of $G$ are {\em vertices}, elements of $E^{\str G}$ are {\em edges}.
A {\em subgraph} of $\str G$ is a graph $\str H=(H,E^{\str H})$ with $H\subseteq G$ and $E^{\str H}\subseteq E^{\str G}$.
The {\em Gaifman graph} of a $\tau$-structure $\str A$ is the graph $\str G(\str A)$ with the same universe $A$ as $\str A$ and 
$(a,a')\in E^{\str G(\str A)}$ if $a,a'$ are distinct and appear together in some tuple $\bar a\in R^{\str A}$ 
for some relation symbol $R\in\tau$.

Let $\str A,\str B$ be $\tau$-structures. A {\em homomorphism from $\str A$ into $\str B$} is a function $h\colon A\to B$ (where $A,B$ 
are the universes of $\str A,\str B$ respectively) such that 
$h(\bar a)\in R^{\str B}$ for every relation symbol $R\in \tau$ and $\bar a\in R^{\str A}$; here for an $r$-tuple $\bar a=(a_1,\ldots, a_{r})\in A^{r}$ 
we write $h(\bar a)$ for the  $r$-tuple $(h(a_1),\ldots, h(a_{r}))\in B^{r}$.  Note $h\colon A\to B$ is a 
homomorphism from $\str A$ into $\str B$ if and only if for every tuple $\bar a$ from $A$ we have
$\atyp(\bar a,\str A)\subseteq \atyp(h(\bar a), \str B)$.
Injective homomorphisms are {\em embeddings}. 
An {\em endomorphism of $\str A$} is a 
homomorphism from $\str A$ to $\str A$. An endomorphism of $\str A$ is {\em trivial} if it is the identity on $A$.

Throughout we mainly stick to the following notational conventions. Classes of structures are denoted by calligraphic letters, structures by boldface letters and their universes by the corresponding italic letter.

\subparagraph*{Path structures}
A {\em path} is a graph $\str G$ isomorphic to $([k],\{(i,j)\mid |i-j|=1\})$ where $k:=|G|$. 
Here, we write $[k]=\{1,\ldots, k\}$ for $k\in\mathbb N,k\ge 1$.
A sequence $g_1,\ldots, g_k$ such that $g_i\mapsto i$ is such an isomorphism is an {\em enumeration of} $\str G$. The vertices $g_1$ and $g_k$ are {\em endpoints}. The path is said to {\em connect} its endpoints and have {\em length} $k-1$ (number of edges). 
If $\str G$ is a graph, then a {\em path in $\str G$} is a subgraph of $\str G$ that is a path.

A \emph{path structure (of vocabulary $\tau$)} $\str P$ is a $\tau$-structure whose Gaifman graph $\str G(\str P)$ is a path. 
An  {\em enumeration of} $\str P$ is an enumeration of $\str G(\str P)$, and by an {\em endpoint of} $\str P$ we mean
 one of $\str G(\str P)$. Note that a path structure has exactly two enumerations. 
 A path structure $\str P$ (of vocabulary $\tau$) is \emph{rooted}
if $\tau$ contains the unary relation symbol $\textit{root}$ such that
$\textit{root}^{\str P}$ is a singleton containing one of the endpoints of $\str G(\str P)$. For a rooted path structure $\str P$ with $|P|=k$, 
by an enumeration of $\str P$ we 
mean an enumeration $p_1,\ldots,p_k$ with $\textit{root}^{\str P}=\{p_1\}$. 
Note that a rooted path structure has exactly one enumeration.
For $i \in [k-1]$, we write $e_i$ for the pair $(p_i, p_{i+1})$;
we refer to the $e_i$ as the \emph{edges} of $\str P$.

\subsection{Parameterized logarithmic space}
We consider  (classical) problems $Q$ as subsets of $\{0,1\}^*$, the set of binary strings.
Our model of computation are Turing machines with a read-only input tape, several work-tapes and a write-only output tape (the head does not move left and writes only $0,1$ and no blank). A Turing machine with oracle $Q$  
additionally has a write-only oracle tape
special states ``yes'', ``no'' and ``?''; 
upon entering ``?'' the content $y$ of the oracle tape is erased, its head placed on the first cell and 
state ``yes'' or ``no'' is entered according to whether $y\in Q$ or not; in such a step the machine is said to {\em query} $y$.

We follow \cite{FlumGrohe06-parameterizedcomplexity} notationally. We view {\em parameterized problems} as pairs $(Q,\kappa)$ where $Q$ is a classical problem 
and $\kappa:\{0,1\}^*\to\mathbb N$ is a {\em parameterization}. We assume that parameterizations are computable in logarithmic space, that is, the binary representation of $\kappa(x)$ is computable from $x$ in space $O(\log |x|)$; here, $|x|$ is the length of $x\in\{0,1\}^*$. 
  
  We exemplify how we present parameterized problems. The parameterized embedding problem associated with a class of structures $\cl A$ is
\npprob{\pemb{$\cl A$}}{A structure $\str A\in\cl A$ and a structure $\str B$}{$|A|$}
{Is there an embedding from $\str A$ into $\str B$?}
The underlying classical problem is $\emb{\cl A}$, the parameterization maps $(\str A,\str B)$ to $|A|$.
 
 The class $\paraL$ consists of those parameterized problems $(Q,\kappa)$ decidable in
  {\em parameterized logarithmic space (with respect to $\kappa$)}, that is, space $f(\kappa(x))+O(\log |x|)$ for 
 some computable $f\colon \mathbb N\to\mathbb N$. Functions computable within this space are {\em pl-computable (with respect to $\kappa$)}.
The class XL consists of those $(Q,\kappa)$ decidable in space $f(\kappa(x))\cdot \log |x|$ for 
 some computable $f\colon\mathbb N\to\mathbb N$. These notions are from~\cite{fellowscai}, our notation follows \cite{para}. 
The class $\PATH$ has been introduced in~\cite{ElberfeldStockhusenTantau15-paramclasses}. It contains those $(Q,\kappa)$ that are accepted by some nondeterministic algorithm (i.e. Turing machine) running in parameterized logarithmic space and which on input $x$ makes at most $f(\kappa(x))\cdot\log |x|$ many nondeterministic steps. 

\begin{theorem}[\cite{ChenMueller15-fineclass}] \label{thm:upperbound} Let $\cl P$ be a decidable class of path structures. Then 
$\pemb{\cl P}\in\PATH$.
\end{theorem}

A {\em pl-reduction} from $(Q,\kappa)$ to $(Q',\kappa')$ is a reduction $R$ from $Q$ to $Q'$ which is pl-computable with respect to $\kappa$
and such that there is a computable $f\colon\mathbb N\to\mathbb N$ such
 that $\kappa'(R(x))\le f(\kappa(x))$ for all $x\in\{0,1\}^*$. 
A {\em pl-Turing reduction} from $(Q,\kappa)$ to $(Q',\kappa')$ is an algorithm with oracle $Q'$ that decides $Q$, 
runs in parameterized logarithmic space with respect to $\kappa$, and has {\em bounded oracle access}: 
there is a computable $f\colon\mathbb N\to\mathbb N$ such that on input $x$ the algorithm only queries
$y$ with $\kappa'(y)\le f(\kappa(x))$.

\begin{remark}\em 
Let $\cl P$ be a decidable class of path structures of vocabulary $\tau$. Any tuple in any relation in a structure in $\cl P$ 
can have at most $2$ distinct components. One can use this observation to give a pl-reduction of  $\pemb{\cl P}$  to $\pemb{\cl P'}$ where 
$\cl P'$ is a decidable class of ``edge-coloured graphs'', i.e.\ path structures 
of some vocabulary $\tau'$ all of whose relation symbols have arity $2$. 
\end{remark}

The following goes back to \cite{ElberfeldStockhusenTantau15-paramclasses}, in the form stated it appears in 
\cite[Theorem~4.7]{ChenMueller15-fineclass}. 

\begin{theorem}\label{thm:ustcon} The following is $\PATH$-complete (with respect to pl-reductions).
\npprob{$p$-$\ustcon$}{A graph $\str G$, vertices $s,t\in G$ and $\ell\in \mathbb N$}{$\ell$}
{Is there a path connecting $s$ and $t$ in $\str G$ of length at most  $\ell$?}
\end{theorem}
This is a parameterized version of the classical {\em undirected $s$-$t$-connectivity problem} $\ustcon$.

\section{Main theorem statement}

Given an ordered pair $e=(a,b)$ we will use the following non-standard notation. When $\ell\in\mathbb N$ we 
define $e^{-\ell}$ to be $(a,b)$ when $\ell$ is even, and to be $(b,a)$ when $\ell$ is odd.

Let $\str P$ be a rooted path structure with  enumeration
$p_1,\ldots, p_k$ where $k:=|P|$. For $d\in\mathbb N$, we say $e_{i}$
is \emph{unfoldable of degree $d$} if $i > d$ and, for each $\ell\in [d]$,
it holds that $\atyp(e_{i}) \not\subseteq \atyp(e_{i-\ell}^{-\ell})$.
Note that all edges are unfoldable of degree 0, and being unfoldable of degree $d$ implies being 
unfoldable of degree $d'$ for all $d'\le d$. By {\em unfoldable} edges we mean edges 
unfoldable of degree 1. 
The {\em unfoldability degree} of $\str P$ is the 
sum $\sum_{i\in[k-1]}d_i$ where $d_i$ is the maximal number
such that $e_i$ is unfoldable of degree $d_i$.

A class of rooted path structures
$\mathcal{P}$  has \emph{bounded unfoldability degree}
if there is a constant $c \in \mathbb N$ such that
every  $\str P \in \mathcal{P}$
has unfoldability degree at most $c$.

\begin{examples}\label{exs:fold}\em  Consider the path structures pictured in the Introduction.
In the structures pictured in (1) and (4), no edge is unfoldable.
In the structures pictured in (2), exactly the edges at even positions are unfoldable; they are unfoldable of degree 1 but not of degree 2. 
In each structure pictured in (3), only the last edge is unfoldable, 
of degree equal to the number of edges in the structure.
\end{examples}

\begin{theorem}[Main]\label{thm:main}
Let $\mathcal{P}$ be a decidable class of rooted path structures of some vocabulary $\tau$.
If $\mathcal{P}$ has bounded unfoldability degree,
then there is a pl-Turing reduction of 
$\pemb{\cl P}$ to the parameterized problem
\npprob{$p$-$\lgsh$}{A graph $\str G$, vertices $s,t\in G$ and $k,\ell\in \mathbb N$ with $k<\ell$}{$\ell$}
{Is it true that $\str G$ contains a path of length at least $\ell$ with endpoint $s$ or  a path of length exactly $k$ connecting
$s$ and $t$?}
Otherwise, $\pemb{\cl P}$ is $\PATH$-complete.
\end{theorem}

We divide the somewhat lengthy proof into lemmas proved in the next section.  Lemma~\ref{lem:hard} together with
 Theorems~\ref{thm:upperbound} and \ref{thm:ustcon} implies 
the first statement of Theorem~\ref{thm:main}. The second is proved as Lemma~\ref{lem:easy}.

\section{Proof of main theorem}

\subsection{Hardness}

\begin{lemma}\label{lem:hard}  Let $\cl P$ be a decidable class of rooted path structures. Assume $\cl P$
does not have bounded unfoldability degree.
Then there exists a pl-reduction from $p$-$\ustcon$ to $\pemb{\cl P}$.
\end{lemma}

\begin{proof} 
Given an instance 
$(\str G,s,t,\ell)$ of $\ustcon$ 
these reductions 
will proceed in two stages according to the characterization of parameterized logarithmic space as so-called {\em logarithmic 
space after a pre-computation}~\cite{FlumGrohe06-parameterizedcomplexity}. 
In the so-called {\em pre-computation} stage, the input parameter $\ell$ is mapped by a computable function to a 
pair $(\str P,X)$ such that $\str P\in\cl P$ and $X$ is a set of $\ell$ many edges of~$\str P$. The space required can be bounded by a computable function of $\ell$. In our case, this first computation exploits special properties of $\cl P$. 

In the second stage, the output $(\str P,\str B)$, an instance 
of $\pemb{\cl P}$ is produced.  This computation is a logarithmic space computation that takes as input $(\str G,s,t,\ell)$ plus
$(\str P,X)$. In our case, this computation is not going to depend on special properties of $\cl P$. 

Recall that $P$ denotes the universe of $\str P$ and $G$ the universe of $\str G$. The $\tau$-structure
$\str B=\str B(\str G,\str P,X,s,t)$
has universe
$B:= G\times P$.
Let $\pi_1$ and $\pi_2$ denote the projections mapping an ordered pair $(g,p)\in B$ to its first resp. second component. 
Roughly speaking, 
we construct $\str B$ in such a way that an embedding $h$ from $\str P$ into $\str B$ 
such that $h^*:=\pi_2\circ h$ is the identity on $P$
yields a path of length at most $|X|=\ell$ in~$\str G$, and also vice-versa (this is formalized in Claims 1 and 3 below).  
We do this by ensuring that, whenever such $h$ 
maps an edge $e$ of $\str P$ to $(b,b')\in B^2$ then the first components of $(b,b')$ are either equal or transverse an edge of $\str G$; 
the latter is ensured to happen only if $e\in X$.
 Proving correctness of our reduction will then amount 
to selecting an appropriate $X$ such that we can prove that $h^*$ is the identity on $P$.

We need some notation. 
Let $k:=|P|$ and let $p_1,\ldots, p_k$ be the enumeration of $\str P$. Recall we write
$e_i$ for the edge $(p_i,p_{i+1})$. Let
$1\le i_1<\cdots < i_\ell<k$ be such that $X=\{e_{i_1},\ldots, e_{i_\ell}\}$.

The structure
 $\str B$ 
interprets the unary relation symbol $\textit{root}\in\tau$ by 
$\textit{root}^{\str B}:=\{(s,p_1)\}$.
To define the  interpretation $R^{\str B}$ of an $r$-ary relation symbol $R\in \tau$ we  
describe an algorithm $\mathbb A$ that given an $r$-tuple $((g_1,q_1),\ldots, (g_r,q_r))\in B^r$ along with $\str G,\str P,X,s,t$ decides whether 
$((g_1,q_1),\ldots, (g_r,q_r))\in R^{\str B}$: 

\begin{enumerate}
\item check that $(q_1,\ldots, q_r)\in R^{\str P}$;
\item compute $i\in[k-1]$ such that every $q_j,j\in[r],$ equals  $p_i$ or $p_{i+1}$;
\item check that there are  $g,g'\in G$ such that every $(g_j,q_j),j\in[r],$  equals $(g,p_i)$ or $(g',p_{i+1})$;
\item if $i+1=k$, then check that $g'=t$;
\item if $i=i_1$, then check that $g=s$;
\item if $e_i\notin X$, then check that $g=g'$;
\item if $e_i\in X$, then  check that $g=g'$ or $(g,g')\in E^{\str G}$;
\item accept.
\end{enumerate}

We understand that the computation is aborted and the algorithm rejects in case one of the 
checks fails. In particular, line 2 is entered only in case $(q_1,\ldots, q_r)\in R^{\str P}$; then the Gaifman graph
$\str G(\str P)$ of $\str P$ contains an edge between any two different components $q_j$'s; since $\str G(\str P)$ is a path, 
the $i$ asked for in line 2 is well-defined. \medskip

The following figure illustrates the construction. Consider the rooted path structure $\str P$ with universe $P=\{p_1,\ldots,p_8\}$ that interprets $\textit{root}$ by $p_1$ (depicted by the filled node) and a binary relation symbol $R$ by the depicted arrows; further consider a directed graph $\str G$ with vertices $G=\{s,t,g,g'\}$ and directed edges again depicted by arrows:\\

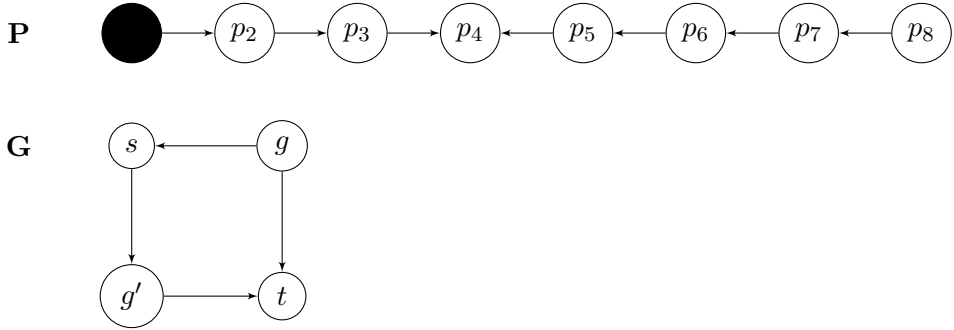
\begin{tikzpicture}

\tikzset{vertex/.style = {shape=circle,draw}}
\tikzset{edge/.style = {->,> = latex'}}


\node (p) at  (0,3.5) {$\str P$};
\node[vertex,fill] (p1) at  (1.5,3.5) {$p_1$};
\node[vertex] (p2) at  (3,3.5) {$p_2$};
\node[vertex] (p3) at  (4.5,3.5) {$p_3$};
\node[vertex] (p4) at  (6,3.5) {$p_4$};
\node[vertex] (p5) at  (7.5,3.5) {$p_5$};
\node[vertex] (p6) at  (9,3.5) {$p_6$};
\node[vertex] (p7) at  (10.5,3.5) {$p_7$};
\node[vertex] (p8) at  (12,3.5) {$p_8$};

\draw[edge] (p1) to (p2);
\draw[edge] (p2) to (p3);
\draw[edge] (p3) to (p4);
\draw[edge] (p5) to (p4);
\draw[edge] (p6) to (p5);
\draw[edge] (p7) to (p6);
\draw[edge] (p8) to (p7);


\node (p) at  (0,2) {$\str G$};
\node[vertex] (s) at (1.5,2){$s$};
\node[vertex] (g) at (3.5,2){$g$};
\node[vertex] (g') at (1.5,0){$g'$};
\node[vertex] (t) at (3.5,0){$t$};

\draw[edge] (s) to (g');
\draw[edge] (g) to (s);
\draw[edge] (g') to (t);
\draw[edge] (g) to (t);

\end{tikzpicture}

The unfoldable edges in $\str P$ are $e_2,e_3,e_5,e_6,e_7$. The $\{\textit{root},R\}$-structure $\str B=\str B(\str G,\str P,\{e_2,e_6\},s,t)$ exemplifies the constuction in case 1 below. It has universe $G\times P$ and looks as follows. We draw a matrix of nodes with rows indexed by $G$ and columns indexed by $P$. The interpretation of $\textit{root}$ is $\{(s,p_1)\}$ and indicated by a filled node. The interpretation of $R$ is by the arrows depicted.\\

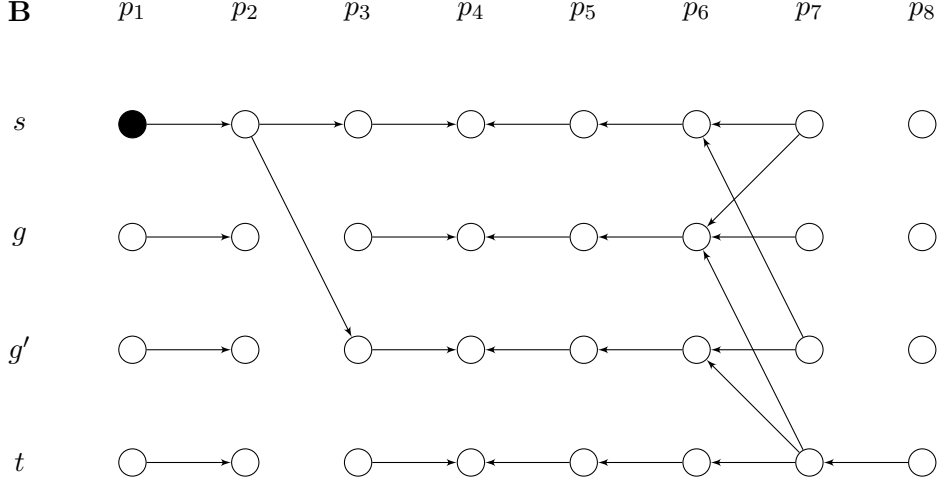
\begin{tikzpicture}

\tikzset{vertex/.style = {shape=circle,draw}}
\tikzset{edge/.style = {->,> = latex'}}


\node (b) at (0,6) {$\str B$};
\node (p1) at  (1.5,6) {$p_1$};
\node (p2) at  (3,6) {$p_2$};
\node (p3) at  (4.5,6) {$p_3$};
\node (p4) at  (6,6) {$p_4$};
\node (p5) at  (7.5,6) {$p_5$};
\node (p6) at  (9,6) {$p_6$};
\node (p7) at  (10.5,6) {$p_7$};
\node (p8) at  (12,6) {$p_8$};
\node (s) at  (0,4.5) {$s$};
\node (g) at  (0,3) {$g$};
\node (g') at  (0,1.5) {$g'$};
\node (t) at  (0,0) {$t$};


\node[vertex,fill] (s1) at  (1.5,4.5) {};
\node[vertex] (g1) at  (1.5,3) {};
\node[vertex] (g'1) at  (1.5,1.5) {};
\node[vertex] (t1) at  (1.5,0) {};
\node[vertex] (s2) at  (3,4.5) {};
\node[vertex] (g2) at  (3,3) {};
\node[vertex] (g'2) at  (3,1.5) {};
\node[vertex] (t2) at  (3,0) {};
\node[vertex] (s3) at  (4.5,4.5) {};
\node[vertex] (g3) at  (4.5,3) {};
\node[vertex] (g'3) at  (4.5,1.5) {};
\node[vertex] (t3) at  (4.5,0) {};
\node[vertex] (s4) at  (6,4.5) {};
\node[vertex] (g4) at  (6,3) {};
\node[vertex] (g'4) at  (6,1.5) {};
\node[vertex] (t4) at  (6,0) {};
\node[vertex] (s5) at  (7.5,4.5) {};
\node[vertex] (g5) at  (7.5,3) {};
\node[vertex] (g'5) at  (7.5,1.5) {};
\node[vertex] (t5) at  (7.5,0) {};
\node[vertex] (s6) at  (9,4.5) {};
\node[vertex] (g6) at  (9,3) {};
\node[vertex] (g'6) at  (9,1.5) {};
\node[vertex] (t6) at  (9,0) {};
\node[vertex] (s7) at  (10.5,4.5) {};
\node[vertex] (g7) at  (10.5,3) {};
\node[vertex] (g'7) at  (10.5,1.5) {};
\node[vertex] (t7) at  (10.5,0) {};
\node[vertex] (s8) at  (12,4.5) {};
\node[vertex] (g8) at  (12,3) {};
\node[vertex] (g'8) at  (12,1.5) {};
\node[vertex] (t8) at  (12,0) {};


\draw[edge] (s1) to (s2);
\draw[edge] (s2) to (s3);
\draw[edge] (s3) to (s4);
\draw[edge] (s5) to (s4);
\draw[edge] (s6) to (s5);
\draw[edge] (s7) to (s6);

\draw[edge] (g1) to (g2);
\draw[edge] (g3) to (g4);
\draw[edge] (g5) to (g4);
\draw[edge] (g6) to (g5);
\draw[edge] (g7) to (g6);

\draw[edge] (g'1) to (g'2);
\draw[edge] (g'3) to (g'4);
\draw[edge] (g'5) to (g'4);
\draw[edge] (g'6) to (g'5);
\draw[edge] (g'7) to (g'6);

\draw[edge] (t1) to (t2);
\draw[edge] (t3) to (t4);
\draw[edge] (t5) to (t4);
\draw[edge] (t6) to (t5);
\draw[edge] (t7) to (t6);
\draw[edge] (t8) to (t7);

\draw[edge] (s2) to (g'3);

\draw[edge] (g'7) to (s6);
\draw[edge] (t7) to (g'6);
\draw[edge] (s7) to (g6);
\draw[edge] (t7) to (g6);

\end{tikzpicture}

\hspace*{2cm}

Continuing with the proof, we make four observations concerning the structure $\str B=\str B(\str G,\str P,X,s,t)$. 

\medskip

\noindent{\em Claim 1:}  If 
$(\str G,s,t,\ell)\in\ustcon$, then $(\str P,\str B)\in\emb{\cl P}$.\medskip

{\em Proof of Claim 1.} 
Assume that $s=g_1,\ldots, g_{\ell'+1}=t$ is (an enumeration of) a length $\ell'\le \ell$ path 
connecting $s$ and $t$ in $\str G$. Define
the sequence $\tilde g_1,\ldots, \tilde g_k$ in $G$ as follows. The first $i_1$ many members $\tilde g_1,\ldots, \tilde g_{i_1}$ 
equal $g_1=s$, the next $i_2-i_1$ members $\tilde g_{i_1+1},\ldots, \tilde g_{i_2}$ equal $g_2$, and so on, the last
$k-i_{\ell'}$ many members $\tilde g_{i_{\ell'}+1},\ldots, \tilde g_k$ equal $g_{\ell'+1}=t$.
Then
\[
(s,p_1)=(\tilde g_1,p_1),(\tilde g_2,p_2), \ldots, (\tilde g_k,p_k)=(t,p_k)
\]
is an enumeration of a copy of $\str P$ in $\str B$, i.e. $p_i\mapsto (\tilde g_i,p_i)$ is an embedding from $\str P$ into $\str B$.
\hfill$\dashv$\medskip

Recall $\str G(\str B)$ denotes the Gaifman graph of $\str B$.\medskip

\noindent{\em Claim 2:} 
Let
$
\refl(E^{\str G}):=E^{\str G}\cup\{(g,g)\mid g\in G\}
$
be the reflexive closure of $E^{\str G}$.
Then
\begin{eqnarray}\label{eq:hom1}
&& \pi_1\text{ is a homomorphism from $\str G(\str B)$ into $(G,\refl(E^{\str G}))$};\\\label{eq:hom12}
&&\text{for all }e\in E^{\str G(\str B)}: \text{ if }\pi_1(e)\in E^{\str G}, \text{ then } \pi_2(e)\in X   \text{ or } \pi_2(e)^{-1}\in X;  \\\label{eq:hom2}
&&\pi_2\text{ is a homomorphism from $\str B$ into $\str P$}.
\end{eqnarray}

\noindent{\em Proof of Claim 2.}
Let $e=((g,p),(g',p'))\in E^{\str G(\str B)}$.
Then
there is $R\in \tau$ and $\bar b\in R^{\str B}$ such that
$(g,p),(g',p')$ both appear in $\bar b$. Since $\mathbb A$ accepts $\bar b$,  by line 2 
there is $i\in[k-1]$ such that $p=p_i,p'=p_{i+1}$ or vice-versa. As one of 
of the checks in line 6 or 7 is carried out, we have $(g,g')\in \refl(E^{\str G})$.
Further, in case $\pi_1(e)=(g,g')\in E^{\str G}$ we have $g\neq g'$, so by line 6 we must then have $e_i\in X$, implying that $\pi_2(e)\in X$ or
$\pi_2(e)^{-1}\in X$.
This shows \eqref{eq:hom1} and \eqref{eq:hom12}. 

Statement \eqref{eq:hom2} is clear: if $\bar b\in R^{\str B}$, then 
$\mathbb A$ accepts $\bar b$, so
$\pi_2(\bar b)\in R^{\str P}$ by line 1.
\hfill$\dashv$\medskip

\noindent{\em Claim 3:}
Assume $h$ is an embedding from $\str P$ into $\str B$. Then $h^*:=\pi_2\circ h$ is an endomorphism of $\str P$; if $h^*$ is trivial, then
$(\str G,s,t,\ell)\in \ustcon$.\medskip

\noindent{\em Proof of Claim 3.} 
The first statement follows from \eqref{eq:hom2}. Assume $h^*$ is trivial, that is, $h^*(p_i)=p_i$ for all $i\in[k]$. 
For $i\in[k]$ let $g_i\in G$ be such that $h(p_i)=(g_i,p_i)$. By \eqref{eq:hom1}, $g_1,\ldots,g_k$ satisfies 
$\pi_1(h(e_i))=(g_i,g_{i+1})\in\refl(E^{\str G})$ for all $i\in [k-1]$.
By \eqref{eq:hom12}, $(g_i,g_{i+1})\in E^{\str G}$ only if $\pi_2(h(e_i))=e_i\in X$. Hence the sequence $g_1,\ldots, g_k$ witnesses that $g_1$ and $g_k$ are connected by a path of length at most $|X|=\ell$ in $\str G$. We are left 
to show $g_1=s$ and $g_k=t$. The former holds as $h(p_1)=(g_1,p_1)\in \textit{root}^{\str B}=\{(s,p_1)\}$. To see
$g_k=t$ note $h(e_{k-1})\in\str G(\str B)$ since $h$ is an embedding. Hence, there are $R\in \tau$ and $\bar b\in R^{\str B}$ with $(g_{k-1},p_{k-1}),(g_k,p_k)$ appearing 
in $\bar b$. Since $\mathbb A$ accepts $\bar b$ we have $g_k=t$ by line 4.
\hfill$\dashv$\medskip

Note that an endomorphism of $\str P$ may fail to be an endomorphism of its Gaifman graph $\str G(\str P)$, for example, $\str P$ could have a constant endomorphism mapping each point $p_i$ to the root $p_1$. However, for endomorphisms of the form as in Claim 3, this can not happen: \medskip

\noindent{\em Claim 4:}
Assume $h$ is an embedding from $\str P$ into $\str B$. Then $h^*:=\pi_2\circ h$ is an endomorphism of $\str G(\str P)$, that is,
for all $i\in[k-1]$ there is $j\in[k-1]$ such that $h^*(e_i)\in\{e_{j},e_j^{-1}\}$.
\medskip

\noindent{\em Proof of Claim 4.} 
By definition of $\mathbb A $ we have
for all $g,g'\in G$ and $ j,j'\in[k]$:
\begin{equation}\label{eq:grid}
\textup{if } ((g,p_j),(g',p_{j'}))\in \str G(\str B), \textup{then } |j-j'|=1.
\end{equation}

For $i\in[k-1]$ choose $g,g',j,j'$ such that $h(e_i)=((g,p_j),(g',p_{j'}))$. 
Note $h$ is also an embedding from $\str G(\str P)$ into $\str G(\str B)$. 
Hence $h(e_i)\in \str G(\str B)$, so $|j-j'|=1$ by \eqref{eq:grid}, i.e. $h^*(e_i)=(p_j,p_{j'})\in\{e_j,e_j^{-1}\}$.
\hfill$\dashv$\medskip

We now exhibit a pl-reduction from $p$-$\ustcon$ to $\pemb{\cl P}$ assuming $\cl P$ does not have bounded unfoldability degree. The assumption implies that $\cl P$ has at least one of the following properties:
\begin{enumerate}\itemsep=0pt
\item[] (Case 1) For every $\ell\in\mathbb N$ there exists $\str P\in\cl P$ such that 
at least $\ell$ many edges of $\str P$ are unfoldable.
\item[] (Case 2) For  every $\ell\in\mathbb N$ there exists $\str P\in\cl P$ such that at least one edge of $\str P$ is unfoldable of degree $\ell$.
\end{enumerate}
We exhibit a reduction as desired in both cases.

\subparagraph*{Case 1} 
Given an 
instance $(\str G,s,t,\ell)$ of $p$-$\ustcon$ the reduction first computes $\str P\in \cl P$ such that $\str P$ contains at least $\ell$ many unfoldable edges.  This can be done by computably enumerating the decidable class~$\cl P$ and testing for each structure output by the enumeration whether it has at least $\ell$ many unfoldable edges or not; the first structure passing the test is $\str P$.

Recall $P$ denotes the universe of $\str P$. Again we write $k:=|P|$ and let $p_1,\ldots,p_k$ be the enumeration of $\str P$. From the input
$\str P$ alone, the reduction computes $\ell$ many unfoldable edges $e_{i_1+1},\ldots, e_{i_\ell+1}$ with $1\le i_1<\cdots < i_\ell<k $. 
Then it computes $X=\{e_{i_1},\ldots,e_{i_\ell}\}$, the set of edges immediately preceeding the~$\ell$ unfoldable ones, 
and outputs $(\str P,\str B)$ for 
$\str B=\str B(\str G,\str P,X,s,t)$. 

We have to show that
\begin{equation}\label{eq:red}
(G,s,t,\ell)\in\ustcon\iff (\str P,\str B)\in \emb{\cl P}.
\end{equation}
Then we are done: as has already been observed, the algorithm $\mathbb A$ runs in logarithmic space, so our reduction
$(G,s,t,\ell)\mapsto(\str P,\str B)$ can be computed in space $f(\ell)+O(\log |G|)$ for some 
computable  $f\colon\mathbb N\to \mathbb N$. The output parameter $k=|P|$ is bounded by (in fact, equal to) a 
computable function of the input parameter~$\ell$. Thus, $(G,s,t,\ell)\mapsto(\str P,\str B)$ defines a pl-reduction from $p$-$\ustcon$ to $\pemb{\cl P}$.

We verify \eqref{eq:red}. The forward direction follows from Claim 1.
Conversely, let $h\colon P\to B$ be an embedding from $\str P$ into $\str B$. By Claim 3
it suffices to show that the endomorphism $h^*:=\pi_2\circ h$ is trivial.

Suppose not and choose the minimal $i\in[k]$ such that 
$h^*(p_i)\neq p_i$. 
Since $h^*$ is an endomorphism (Claim 3) and $\textit{root}^{\str P}=\{p_1\}$ we have $h^*(p_1)=p_1$, so $i>1$.
Since $h^*(e_1)$ is an edge of $\str G(\str P)$  (Claim 4), we have $h^*(e_1)=e_1$, so $i>2$.
Since $h^*(p_{i-1})=p_{i-1}$ and $h^*(e_{i-1})$ is an edge of $\str G(\str P)$ (Claim 4),
we have $h^*(e_{i-1})=e_i$ or $h^*(e_{i-1})=e^{-1}_{i-2}$. As $h^*(p_i)\neq p_i$, we have $h^*(e_{i-1})=e^{-1}_{i-2}$, that is,
\[
h(e_{i-1})=((g,p_{i-1}),(g',p_{i-2}))
\] 
for certain $g,g'\in G$. Further, $\atyp(e_{i-1})\subseteq  \atyp(e_{i-2}^{-1})$ since $h^*$ is an endomorphism of $\str P$.
Thus, $e_{i-1}$ is not unfoldable and therefore $e_{i-2}\notin X$. Since $e_{i-1}\in \str G(\str P)$ there 
exists $R\in \tau$ and a tuple $\bar q\in R^{\str P} $ such that both $p_{i-1}$ and $p_{i}$ 
appear in $\bar q$. The image $h(\bar q)$  contains $(g,p_{i-1})$ and $(g',p_{i-2})$. 
Since $h(\bar q)\in R^{\str B}$ and $e_{i-2}\notin X$, algorithm  $\mathbb A$ accepts carrying out the check in line 6, so $g=g'$.

For the $g''\in G$ such that $h(p_{i-2})=(g'',p_{i-2})$, we have $h(e_{i-2})=((g'',p_{i-2}),(g,p_{i-1}))$. As above we see
that also $(g'',p_{i-2}),(g,p_{i-1})$ appear in some tuple in some relation from $\str B$, and hence $g''=g$.

Thus, $h(p_i)=h(p_{i-2})=(g,p_{i-2})$ and $h$ is not injective, a contradiction.

\subparagraph*{Case 2} In this case the 
reduction maps  an instance $(\str G,s,t,\ell)$ to $(\str P,\str B)$ for $\str B=\str B(\str G,\str P,X,s,t)$
where $\str P\in\cl P$ has an edge that is unfoldable of degree 
 $\ell$ and $X$ is the set of $\ell$ many  edges preceeding this edge.
More precisely, for $k:=|P|$ let $p_1,\ldots,p_k$ be the enumeration of $\str P$ and let $e_i=(p_i,p_{i+1})$ 
be unfoldable of degree $\ell$; then  $X$ is $\{e_{i-1},\ldots, e_{i-\ell}\}$.
Such $(\str P,X)$ can be computed from $\ell$.
We have to show:
\begin{equation}\label{eq:red2}
(G,s,t,\ell)\in\ustcon \Longleftrightarrow (\str P,\str B)\in \emb{\cl P}.
\end{equation}

By Claim 1 it suffices to prove the backward direction in  \eqref{eq:red2}.
So assume $h$ is an embedding from $\str P$ into~$\str B$. By Claim 3 it suffices to show that the endomorphism 
$h^*:=\pi_2\circ h$ is trivial.

We first show that $h^*(p_j)=p_j$ for all $j\in [i-\ell]$ and in fact $h(p_j)=(s,p_j)$. 
Since $h(p_1)\in \textit{root}^{\str B}=\{(s,p_1)\}$ this holds for $j=1$. Inductively, assuming 
$1\le j<i-\ell$ and $h(p_j)=(s,p_j)$ we show $g=s,q=p_{i+1}$ for $g,q$ such that $h(p_{j+1})=(g,q)$. 

Since $h$ is an embedding, $h(e_j)\in \str G(\str B)$, so $\pi_1(h(e_j))\in\refl(E^{\str G})$ by \eqref{eq:hom1}. 
Neither $(p_j,q)$ nor $(q,p_{j})$ are in $ X$ since $j<i-\ell$.
By \eqref{eq:hom12} we have
$\pi_1(h(e_j))\in\refl(E^{\str G})\setminus E^{\str G}$, i.e.\ $s=g$.
Now,  $h^*(e_j)=(p_j,q)$ is an edge of $\str G(\str P)$ (Claim 4), so equals $e_{j}$ or $e_{j-1}^{-1}$. In the first case, $q=p_{j+1}$ 
and we are done. The other case is impossible:  it implies $h(p_{j+1})=(s,p_{j-1})=h(p_{j-1})$, contradicting  
the injectivity of~$h$.

Each of the $\ell$ many points   $p_{i-\ell+1},\ldots , p_{i}$ is mapped by $h^*$
to some $p_{i-\ell+j}$ with $j\ge 1$.
Otherwise, there is $j\in[\ell-1]$ such that $h^*(p_{i-\ell+j})=p_{i-\ell+1}$ and
$h^*(p_{i-\ell+j+1})=p_{i-\ell}$, that is, $h^*(e_{i-\ell+j})=e_{i-\ell}^{-1}$.
Write  $h(e_{i-\ell+j})=((g,p_{i-\ell+1}),(g',p_{i-\ell}))$ for certain $g,g'\in G$.
Since this is an edge  in $\str G(\str B)$
there are $R\in\tau$ and $\bar b\in R^{\str B}$ such that $(g,p_{i-\ell+1}),(g',p_{i-\ell})$ appear in $\bar b$. Since $\mathbb A$ accepts
$\bar b$, it follows from line 5 that $g'=s$ (note $i-\ell$ is the smallest index of an edge in $X$). 
Then $h(p_{i-\ell+j+1})=(s,p_{i-\ell})=h(p_{i-\ell})$, contradicting the injectivity of $h$.

In particular, $h^*(p_{i-\ell+1})=p_{i-\ell+1}$ as otherwise $h^*(e_{i-\ell})=(p_{i-\ell},p_{i-\ell+j})$ for some $j> 1$ would not 
be an edge of $\str G(\str P)$ (contradicting Claim 4).

Observe that in the graph $\str G(\str B)$ each vertex  $(g,p_{i-\ell+j})$ with $j\ge 1$ has only 
neighbors in $G\times \{p_{i-\ell+j+1},p_{i-\ell+j-1}\}$ (cf.~\eqref{eq:grid}).
This implies that in the length $\ell$ sequence of points $h^*(p_{i-\ell+1}),\ldots, h^*(p_{i})$, the index increases or decreases 
by 1 in every step. 
It implies further that for none of these points $h^*$ changes the parity of the index -- more precisely:
if $\nu,\mu$ are indices $>i-\ell$ and $\le i$ such that $h^*(p_\mu)=p_\nu$ then the parities of $\mu$ and $\nu$ are equal. Indeed, 
this follows easily by induction on $\mu$ with $i-\ell<\mu\le i$: for the base case $\mu=i-\ell+1$ note $h^*(p_{i-\ell+1})=p_{i-\ell+1}$; 
the induction step follows from the previous observation that if $\mu$ increases by 1, then the index of $h^*(p_\mu)$ increases or decreases by 1, so changes  its parity.

We next show that the index increases by 1 in every step, that is,
$h^*(p_{i-\ell+j})=p_{i-\ell+j}$ for all $j\in[\ell]$.
Otherwise, $h^*(p_i)=p_{\nu}$ for 
some  $i-\ell<\nu <i$. Then $h^*(e_{i})$ equals $e_{\nu}$ or $e^{-1}_{\nu-1}$. We show both is impossible.
Since the parity of $\nu$ equals that of $i$ we have that $d:=i-\nu$ is even. Thus, by the assumption that $e_i$ is unfoldable of degree $\ell$,
\begin{eqnarray*}
&&\atyp(e_i)\not\subseteq\atyp(e_{i-d}^{-d})=\atyp(e_\nu),\textup{ and}\\
&&\atyp(e_i)\not\subseteq\atyp(e_{i-(d+1)}^{-(d+1)})=\atyp(e_{\nu-1}^{-1}).
\end{eqnarray*}
 
We now know $h^*(p_j)=p_j$ for all $j\in[i]$. 
Then $h^*(p_{i+1})=p_{i+1}$: otherwise, by Claim 4, $h^*(p_{i+1})=p_{i-1}$, so
$h^*(e_i)=e_{i-1}^{-1}$ and $\atyp(e_i)\subseteq \atyp(e_{i-1}^{-1})$ contradicting the unfoldability of $e_i$. It now follows that 
$h^*(p_j)=p_j$ for all $i<j\le k$, and indeed $h(p_j)=(g,p_j)$ for $g:=\pi_1(h(p_{i+1}))$. This is verified similarly 
as for the points $p_j$ with $j\in[i-\ell]$.
\end{proof}

\subsection{Upper bound} 


\begin{lemma}\label{lem:easy} Let $\cl P$ be a decidable class of rooted path structures. 
Suppose $\cl P$ has bounded unfoldability degree.
Then there is a pl-Turing reduction from
$\pemb{\cl P}$ to $p$-$\lgsh$.
\end{lemma}

The proof uses color-coding, namely, we shall rely on the following lemma \cite[page~349]{FlumGrohe06-parameterizedcomplexity}.

\begin{lemma}\label{lem:cc}
For every sufficiently large $n\in \mathbb N$, it holds that for all $k\le n$ and
for every $k$-element subset $X$ of $[n]$, there exists a prime $p <
k^2\cdot \log n$ and $q< p$ such that the function $h_{p,q}: [n]\to \{0,
\ldots, k^2 -1\}$ given by 
$
h_{p,q}(m):= (q\cdot m \mod p) \mod
k^2
$ 
is injective on $X$.
\end{lemma}

\begin{proof}(of Lemma~\ref{lem:easy}) 
Choose constants $c,d\in\mathbb N$ such that every structure $\str P\in \mathcal{P}$ has 
at most $c$ many unfoldable edges, and whenever an edge thereof is unfoldable of degree $g$,
it holds that $g \leq d$.

Call an edge of a rooted path structure $\str P$ {\em critical} if its atomic type is different from the atomic type of the inverse of the edge preceding it. More precisely, let $k:=|P|$ and $p_1,\ldots, p_k$ be the enumeration of $\str P$; recall we write $e_i$ for the edge $(p_i,p_{i+1})$ of $\str P$. For $i\in[k]$ the edge $e_i$ is {\em critical}
if $i>1$ and $\atyp(e_i)\neq\atyp(e_{i-1}^{-1})$. This means $e_i$ is either unfoldable or
$\atyp(e_i)\subsetneq\atyp(e_{i-1}^{-1})$. Observe that $e_2$ is critical since $p_3\notin\textit{root}^{\str P}$, so 
$\atyp(e_2)\neq\atyp(e_1^{-1})$. Thus every rooted path structure of size at least 3 (i.e. $k\ge 3$)
has at least one critical edge.

For example, consider again the rooted path structures depicted in the Introduction. The critical edges of the structures in (1), (2) and (3) are
the unfoldable ones (cf.~Examples~\ref{exs:fold}). The critical edges of the structures in (4) are the second one (unfoldable) and the first non-symmetric one (not unfoldable). Note that the structures in (3) and (4) have both exactly 2 critical edges. We have already seen that the first class has an associated embedding problem which is PATH-hard, and shall now see  that
the second one has an associated embedding problem which is pl-Turing reducible to $p$-$\lgsh$.

For $C\in\mathbb N$ let $\cl P(C)$ be the class of all $\str P\in \cl P$ having at most $C$ many critical edges.
For every $C\in\mathbb N$ we are going to define a pl-Turing reduction from $\pemb{\cl P(C)}$
 to $p$-$\lgsh$.
 
This suffices to prove the lemma:

\medskip

\noindent{\em Claim 5:}
There exists $C\in\mathbb N$ such that  $\cl P=\cl P(C)$. 
\medskip

\noindent{\em Proof of Claim 5.} Let $t$ denote the number of atomic $\tau$-formulas in two variables.
Let $\str P\in\cl P$ and let $e$ be an unfoldable edge of $\str P$ or the last edge. Let $e'$ be the unfoldable edge preceding $e$ (according to the enumeration of $\str P$); if there is no such unfoldable edge, let $e'$ be the first edge.
Then 
there are $|\atyp(e')|\le t$ many critical edges between $e$ and $e'$. In total,  $\str P$ has at most $C:=c+t(c+1)$ many critical edges.
\hfill$\dashv$\medskip

Let $\str P\in\cl P$ with $k:=|P|> 2+cd+d$ and
enumeration $p_1,\ldots, p_k$. 
Then there exists $a$ with $1 < a < k-d$ such that
none of $e_{a+1}, \ldots, e_{a+d}$ 
are unfoldable.
Fix $a$ to be the minimum such value.
As $\str P$ has at most $c$ unfoldable edges, we have
%
\begin{equation}\label{eq:a}
1<a\le 2+cd.
\end{equation}
The following claim explains our interest in the number $a$. 

\medskip

\noindent{\em Claim 6:}
For all $i\le k-a$ we have $\atyp(e_{a+i})\subseteq\atyp(e_a^{-i})$.
\medskip

\noindent{\em Proof of Claim 6.} 
Suppose not. Choose $i$ minimal such that $\atyp(e_{a+i})\not\subseteq\atyp(e_a^{-i})$. Then for all $j\in [i]$ we have
$\atyp(e_{a+i-j})\subseteq \atyp(e_a^{-(i-j)})$, which implies $\atyp(e_{a+i-j}^{-j})\subseteq \atyp(e_a^{-i})$, and thus
$\atyp(e_{a+i})\not\subseteq\atyp(e_{a+i-j}^{-j})$. This means $e_{a+i}$ is unfoldable of degree $i$, so $i\le d$. 
This contradicts the choice of $a$. 
\hfill$\dashv$\medskip

Our algorithm $\mathbb A(C)$ is going to be recursive, and the depth of the recursion bounded by a constant. It recurses on the parts $\str P\uparrow i$ and $\str P\downarrow i$ of $\str P$ obtained by ``cutting $\str P$ at point $i$.'' These structures are defined for $i\in[k]$ as follows: $\str P\downarrow i$ is the substructure of $\str P$
induced on $\{p_1,\ldots,p_i\}$, and $\str P\uparrow i$ is obtained from the substructure of $\str P$ induced on $\{p_i,\ldots, p_{k}\}$ by declaring $p_i$ the new root, i.e.\ interpreting $\textit{root}\in \tau$ by~$\{p_i\}$.
Without loss of generality we can assume that if $\cl P$ contains $\str P$, then it also contains  all these structures $\str P\uparrow i,\str P\downarrow i$. If this would not be the case, we could add all these structures to $\cl P$ and observe that the resulting bigger class still satisfies the assumptions of the lemma.

Let $(\str P,\str B)$ be an instance of $\pemb{\cl P(C)}$ and $p_1,\ldots, p_k$ 
be the enumeration of $\str P$. 
If $k:=|P|\le 2+cd+d$ or if $n:=|B|$ is not sufficiently large 
in the sense of Lemma~\ref{lem:cc} or if $n\le k$, then $\mathbb A(C)$ uses ``brute force'', that is, it 
runs an XL algorithm for $\pemb{\cl P(C)}$. So assume 
$2+cd+d< k=|P|<|B|=n$ and~$n$ is sufficiently large. 
Then the algorithm $\mathbb A(C)$ computes the number $a$. There are two cases.

\subparagraph{Case} $a>2$. In this case,
 $\mathbb A(C)$ 
loops through all functions $h\colon \{p_1,\ldots,p_{a}\}\to B$. For every such~$h$ the algorithm checks whether
it is an embedding of $\str P\downarrow a$ into~$\str B$. If so, $\mathbb A(C)$ recurses on the instance $(\str P\uparrow a,\str B_h)$
where
$\str B_h$ is obtained from~$\str B$ and $h$ as follows: take the substructure induced in~$\str B$ with universe~$B\setminus \{h(p_1),\ldots, h(p_{a-1})\}$ and change the interpretation of $\textit{root}$ to~$\{h(p_a)\}$.
If the recursive call returns accepting, then $\mathbb A(C)$ halts and accepts. 

By \eqref{eq:a}, each $h$ can be stored using 
$O(\log n)$ bits. Hence, in the current case $\mathbb A$ uses $O(\log n)$ space plus the space required by the recursive calls.

\subparagraph{Case} $a=2$. In this  case we shall use the following construction. It is similar to a construction used
 in  \cite[Theorem~18]{ChenMueller14-paramlogspace}.
We can assume that the $\tau$-structure $\str B$ has universe $B=[n]$ for some $n>2$  sufficiently large in the sense of
 Lemma~\ref{lem:cc}. Using the notation from this lemma, set 
\[
F:= \big\{g \circ h_{p,q} \mid g\colon \big\{0, \ldots, (k-1)^2-1\big\} \to \{2,\ldots,k\}
 \text{ and } q< p< (k-1)^2 \log n\big\}.
\]
For $f\in F$ and $b_1\in B$ define the following graph $\str G(f,b_1)$. Its vertices are $B\setminus \{b_1\}$. Its set of edges is
the symmetric closure of the set of all $(b,b')\in (B\setminus\{b_1\})^2$ such that
\begin{itemize}
\item[--]  $f(b)+1=f(b')\text{ and } \atyp(e^{-f(b)}_{2},\str P)\subseteq\atyp((b,b'),\str B)$; and
\item[--] if $f(b)=2$ then $\atyp(e_{1},\str P)\subseteq\atyp((b_1,b),\str B)$.
\end{itemize}
{\em Informally}, the idea is as follows. The first condition means to put an edge between all strongest (cf.~Claim~7) $\atyp(e_2)$-edges in $\str B$, but only  between vertices with neighboring $f$-colours; the
second condition ensures that a vertex with colour $2$ is isolated unless it is a $\atyp(e_1)$-successor of $b_1$ in $\str B$.

The following two claims pinpoint the properties we need this construction to have.\medskip

\noindent{\em Claim 7:}
Let $b_1\in B,f\in F$ and $s\in f^{-1}(2)$. If $\str G(f,b_1)$ contains a length $k-2$ path with endpoint $s$, then there is an embedding from 
$\str P$ into $\str B$.
\medskip

\noindent{\em Proof of Claim 7.}  Let $s=b_2,b_3,\ldots,b_{k}$ enumerate a path in $\str G(f,b_1)$,
Clearly, $\atyp(e_1,\str P)\subseteq \atyp((b_1,s),\str B)$. 
For $1<j\le k-1$ we have $|f(b_j)-f(b_{j+1})|=1$. Since $f(s)=f(b_2)=2$ it follows 
that $f(b_j)$ has the same parity as $j$. Thus, if $f(b_j)<f(b_{j+1})$, then 
\[
\atyp(e_j)\subseteq\atyp(e_2^{-j})= \atyp(e_2^{-f(b_j)})\subseteq\atyp((b_j,b_{j+1}),\str B),
\]
where the first inclusion holds by Claim 6; if $f(b_{j+1})<f(b_{j})$, then similarly
\[
\atyp(e_j)\subseteq\atyp(e_2^{-j})= \atyp((e_2^{-f(b_{j+1})})^{-1})\subseteq\atyp((b_{j+1},b_{j})^{-1},\str B)=\atyp((b_j,b_{j+1}),\str B).
\]
 It follows that $p_j\mapsto b_j$ defines an embedding of $\str P$ into $\str B$. \hfill$\dashv$\medskip

\noindent{\em Claim 8:}
Let $b_1\in B$. The following are equivalent.
\begin{enumerate}
\item[(a)] There exists $f\in F$ such that $\str G(f,b_1)$ contains a length $k-2$ path with one endpoint in $f^{-1}(2)$.
\item[(b)] There 
exists $g\in F$ such that $\str G(g,b_1)$ contains a path with one endpoint in $g^{-1}(2)$ and the other endpoint in $g^{-1}(k)$.
\end{enumerate}

\noindent{\em Proof of Claim 8.} (b) implies (a) because any path in $\str G(g,b_1)$ connecting points in $g^{-1}(2)$ and $g^{-1}(k)$ has 
length at least~$k-2$. Conversely,
suppose $b_2,\ldots, b_k$ enumerates  a path in~$\str G(f,b_1)$. 
By Lemma~\ref{lem:cc} there exists $g\in F$ such that $g(b_j)=j$ for all $2\le j\le k$. Then $b_2,\ldots, b_k$ enumerates  a path in~$\str G(g,b_1)$ with $b_2\in g^{-1}(2)$ and $b_k\in g^{-1}(k)$.
 \hfill$\dashv$\medskip

We now describe how algorithm $\mathbb A(C)$ works in the current case $a=2$. 
It first loops through all tuples $(b_1,p,q,g,s,t)$ such that
$b_1\in B, q< p< (k-1)^2 \log n, g\colon \big\{0, \ldots, (k-1)^2-1\big\} \to \{2,\ldots,k\}$ and $s,t\in B\setminus\{b_1\}$.
For each such tuple $\mathbb A$ first checks that $g(h_{p,q}(s))=2$ and $g(h_{p,q}(t))=k$. 
Second it checks whether in $\str G(g\circ h_{p,q},b_1)$ 
the vertex $s$ is connected by some path to $t$. If one such check is positive, the algorithm stops and accepts.

It follows from Claims 7 and 8, that if $\mathbb A$ accepts here then indeed $(\str P,\str B)\in \pemb{\cl P(C)}$. 
A tuple $(b_1,p,q,g,s,t)$ can be stored using $O(k\log k+\log n)$ bits. And the graph
$\str G(g\circ h_{p,q},b_1)$ is  pl-computable from the input. The second check can be done using  Reingold's
logarithmic space algorithm \cite{Reingold08-conn-in-logspace}. We thus see that this loop can be implemented within the allowed space.

If this first loop did not cause $\mathbb A(C)$ to accept, then $\mathbb A(C)$ recurses as follows. It computes
the index of the first critical edge after $e_2$, i.e. it computes the minimal  
$2<i\le k-1$ such that $e_{i+1}$ is critical. If there is no critical edge after $e_2$, then $\mathbb A(C)$ sets $i:=k-1$.
In both cases we have 
for all $2\le j\le i$: 
\begin{equation}\label{eq:aalt}
\atyp(e_{j})=\atyp(e_2^{-j}).
\end{equation}

Then $\mathbb A(C)$ loops a second time
through all tuples  $(b_1,p,q,g,s,t)$ as before. It first checks
 that $g(h_{p,q}(s))=2$ and $g(h_{p,q}(t))=i$.
Then it queries the
 oracle whether 
 \[
 (\str G(g\circ h_{p,q},b_1), s,t,i-2,k-2)\in \lgsh.
 \] 
 If the oracle answers ``no'', the algorithm considers the next tuple. If the oracle answers ``yes'', then $\mathbb A$ recurses on 
 $(\str P\uparrow i,\str B')$ where $\str B'$ is obtained as follows. Take the induced substructure of $\str B$ with universe 
 $
 B':=\{b\in B\mid i< g\circ h_{p,q}(b)\le k\}\cup\{t\}
 $ 
 and change the interpretation of $\textit{root}$ to $\{t\}$.
 If the recursive call returns accepting, then $\mathbb A(C)$ halts and accepts.

 Note the instance $\mathbb A$ recurses to is  pl-computable from the input. So also in this second loop, $\mathbb A(C)$ uses parameterized logarithmic space plus the space required by the recursive calls. 
 We argue for correctness: if $\mathbb A(C)$ accepts here, then $(\str G(g\circ h_{p,q},b_1), s,t,i-2,k-2)\in \lgsh$. By the fact, that $\mathbb A(C)$ entered the second loop, it follows that the statement Claim 8 (b) is false. Then the statement Claim 8 (a) is false. This implies that $\str G(g\circ h_{p,q},b_1)$ does not contain a length $k-2$ path with endpoint $s$. Therefore,
$(\str G(g\circ h_{p,q},b_1), s,t,i-2,k-2)\in\lgsh$ implies that there exists a path of length exactly $i-2$ from $s$ to $t$ in $\str G(g\circ h_{p,q},b_1)$. 

Let
$s=b_2,\ldots, b_i=t$ enumerate such a path. Then 
\begin{equation}\label{eq:col}
g(h_{p,q}(b_j))=j
\end{equation}
for all $2\le j\le i$.
Assuming inductively that the recursive call accepts correctly, we 
have an embedding $h$ from $\str P\uparrow i$ into $\str B'$. Since $\textit{root}^{\str B'}=\{t\}$ 
we have $h(p_i)=t=b_i$. Hence, $g(h_{p,q}(h(p_j)))>i$ for all $i<j\le k$, and hence $h(p_j)\neq b_{j'}$ for all $2\le j'<i<j\le k$ (by \eqref{eq:col}).
It follows that the following function $h'$ is injective:
map $p_j$ to $b_j$ for $1\le j\le i$, and map $p_j$ to $h(p_j)$ for $i<j\le k$.
%
Moreover, $h'$ is an embedding from $\str P$ into $\str B$: for $2\le j< i$ we have
$h'(e_j)=(b_j,b_{j+1})$ and
\[
\atyp(e_j)=\atyp(e_2^{-j})\subseteq\atyp((b_j,b_{j+1}),\str B),
\]
where the equality follows from \eqref{eq:aalt},  and the inclusion from $(b_j,b_{j+1})$ being an
 edge in $\str G(g\circ h_{p,q},b_1)$ and \eqref{eq:col}. We leave it to the reader to check
$\atyp(e_j)\subseteq \atyp(h'(e_j),\str B)$ if $j=1$ or $i\le j\le k-1$.

We have argued that in all cases when $\mathbb A(C)$ accepts it does correctly so. Conversely, it is routine to verify, using Lemma~\ref{lem:cc}, that $\mathbb A(C)$ accepts if $(\str P,\str B)\in\emb{\cl P(C)}$. We conclude that  $\mathbb A(C)$ decides 
$p$-$\emb{\cl P(C)}$.
Obviously, the oracle access is bounded. We are thus left to analyze the space complexity of $\mathbb  A(C)$. We already argued that
$\mathbb A$ uses parameterized logarithmic space plus the space needed for the recursive calls.
It is thus  sufficient to check that the depth of the recursion is bounded by a constant. 

In all cases $\mathbb A(C)$ recurses on $\str P\uparrow i$ for some $i>2$.
The number of critical edges of $\str P\uparrow i$ is less than or equal to the number of critical edges in $\str P$. The same holds for the number of unfoldable edges. 
Indeed, each critical edge of $\str P\uparrow i$ besides possibly its second one, is also critical in 
$\str P$. And clearly every edge unfoldable in $\str P\uparrow i$ is also unfoldable in $\str P$.

In Case $a>2$, the algorithm recurses on $\str P\uparrow a$. That 
 $a>2$ implies that the first $a-1$ edges in $\str P$ contain at least one unfoldable edge. It follows that
$\str P\uparrow a$ has less unfoldable edges than $\str P$. 

In Case $a=2$,  $\mathbb A(C)$ recurses on $\str P\uparrow i$ for a certain $i>2$. If this $i$ equals $k-1$, then
$\str P\uparrow i$ has only one edge and the recursive call uses ``brute force'' without any further recursion. If $i<k-1$, then
$e_{i+1}$ is a critical edge. This becomes the second edge in $\str P\uparrow i$. Since the critical $e_2$ is not present in
$\str P\uparrow i$, this structure  has strictly less critical edges than $\str P$.

It follows that in each recursive call either the number of unfoldable edges drops or ``brute force'' is applied or 
the number of critical edges drops. We conclude that the recursion depth is at most~$c+C+1$.
\end{proof}

%

\section{Dichotomy for rooted oriented paths}

Let $\tau$ be the vocabulary $\{ \textit{root}, E \}$
where $\textit{root}$ is a unary relation symbol and $E$
is a binary relation symbol.
Let us say that a structure $\str P$ over $\tau$
is a \emph{rooted oriented path} if it is a rooted path structure
with enumeration $p_1, \ldots, p_k$ 
such that, for each $i \in [k-1]$, 
exactly one of the two pairs $(p_i, p_{i+1})$,
$(p_{i+1}, p_i)$ is contained in $E^{\str P}$; and,
no other pairs are in $E^{\str P}$, in particular $\str P$ contains no loops (i.e.\ $E^{\str P}$ is irreflexive). 
For such a structure and $C \geq 1$, let us say that 
the structure has a \emph{$C$-alternating tail} if,
for each $i \geq C$, the edge $e_i$
is foldable
if it exists (that is, if $i+1 \leq k$). This means, that the edges $e_C,e_{C+1},\ldots$ alternate in direction. For example,\medskip

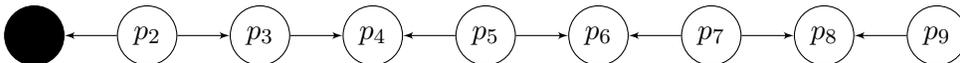
\begin{tikzpicture}

\tikzset{vertex/.style = {shape=circle,draw}}
\tikzset{edge/.style = {->,> = latex'}}

\node[vertex,fill] (p1) at  (1.5,0) {$p_1$};
\node[vertex] (p2) at  (3,0) {$p_2$};
\node[vertex] (p3) at  (4.5,0) {$p_3$};
\node[vertex] (p4) at  (6,0) {$p_4$};
\node[vertex] (p5) at  (7.5,0) {$p_5$};
\node[vertex] (p6) at  (9,0) {$p_6$};
\node[vertex] (p7) at  (10.5,0) {$p_7$};
\node[vertex] (p8) at  (12,0) {$p_8$};
\node[vertex] (p9) at  (13.5,0) {$p_9$};

\draw[edge] (p2) to (p1);
\draw[edge] (p2) to (p3);
\draw[edge] (p3) to (p4);
\draw[edge] (p5) to (p4);
\draw[edge] (p5) to (p6);
\draw[edge] (p7) to (p6);
\draw[edge] (p7) to (p8);
\draw[edge] (p9) to (p8);

\end{tikzpicture}

\noindent pictures a rooted oriented path with a $4$-alternating tail. \medskip

We establish the following dichotomy theorem.

\begin{theorem}
\label{thm:dichotomy-rooted-oriented-paths}
Let ${\cl A}$ be a decidable class of rooted oriented paths.
If there exists $C \geq 1$ such that each structure in ${\cl A}$
has a $C$-alternating tail, then 
$\pemb{\cl A}$ is in $\paraL$.
Otherwise, $\pemb{\cl A}$ is $\PATH$-complete.
\end{theorem}

\begin{proof} 
For the hardness result, assume that there exists no constant $C$
with the described property. By Theorem~\ref{thm:main} is suffices to show that
for each $c \geq 1$ 
there exists  $\str P\in \cl P$ of 
unfoldability degree at least $c-1$. Indeed, if $\str P\in\cl P$ does not have a $c$-alternating tail, 
then the  last unfoldable edge $e_i$ in $\str P$ satisfies $i\ge c$.
Let $e_{i_1},\ldots,e_{i_r}$ with $i_1<\cdots<i_r=i$ list the unfoldable edges.
Set $i_0:=1$ and observe that $e_{i_j}$ is unfoldable of degree $i_j - i_{j-1}$. It follows that the unfoldability 
degree of $\str P$ is at least $i - 1 \ge c - 1$.

Now assume that there exists a constant $C\ge 1$ such that each $\str P\in\cl A$ has a $C$-alternating tail. We have to find an algorithm 
deciding $\pemb{\cl A}$ in parameterized logarithmic space. This algorithm is akin to the oracle algorithm constructed in the proof 
of Lemma~\ref{lem:easy}. 

Let $(\str P,\str B)$ be an instance of  $\pemb{\cl A}$ and $p_1,\ldots, p_k$ be the enumeration of $\str P$. We assume $k>C$ and
 $B=[n]$ for some $n>2$  sufficiently large in the sense of Lemma~\ref{lem:cc}. Set
\[
F':= \Big\{g \circ h_{p,q} \mid g\colon \big\{0, \ldots, (k-C+1)^2-1\big\} \to \{C,\ldots,k\}
 \text{ and } q< p< (k-C+1)^2 \log n\Big\}.
\]

In the following we understand that $f$ ranges over $F'$ and $h$ ranges over the set of functions from $\{p_1,\ldots, p_C\}$ to $ B$. 
The graph $\str G(f,h)$ has vertices $B\setminus \{h(p_1),\ldots, h(p_{C-1})\}$. 
Its set of edges is the symmetric closure of the set of those $(b,b')\in E^{\str B}$ which satisfy
\begin{itemize}\itemsep=0pt
\item[--] $|f(b)-f(b')|=1$, $f(b)\ge C,f(b')\ge C$ and $(p_{f(b)},p_{f(b')})\in E^{\str P}$;
\item[--] if $f(b)=C$ then $b=h(p_C)$;
\item[--] if $f(b')=C$ then $b'=h(p_C)$.
\end{itemize}
Then the following are equivalent.
\begin{enumerate}\itemsep=0pt
\item[(a)] There are
$f,h$ such that $h$ is an embedding of $\str P\downarrow C$ into $\str B$ and 
$\str G(f,h)$ contains a path with one endpoint in $f^{-1}(C)$ and the other endpoint in $f^{-1}(k)$.
\item[(b)] There are
$f,h$ such that $h$ is an embedding of $\str P\downarrow C$ into $\str B$ and 
$\str G(f,h)$ contains a length $k-C$ path with one endpoint in $f^{-1}(C)$.
\item[(c)] $(\str P,\str B)\in\emb{\cl A}$.
\end{enumerate}
Indeed, that (a) and (b) are equivalent is seen similarly as Claim 8. That (b) implies (c) is seen 
similarly as Claim 7. Finally, that (c) implies (a) is easy to see using Lemma~\ref{lem:cc}.

Algorithm $\mathbb B$ on $(\str P,\str B)$ checks that $|P|>C$ and $|B|\ge 2$ 
is sufficiently large in the sense of Lemma~\ref{lem:cc}. If this is not the case, $\mathbb B$ runs an XL algorithm for $\pemb{\cl A}$.
Otherwise $\mathbb B$ checks statement (a). This can be implemented within the allowed space similarly 
as explained for the second loop of the algorithm constructed in the proof of 
Lemma~\ref{lem:easy}. 
\end{proof}

\section{The long-short path problem is unavoidable}

We show here that,
up to pl-reduction, 
there is an embedding problem
equivalent to the long-short path problem.

\begin{theorem}
\label{thm:long-short-as-embedding-problem}
Let $\tau$ be the vocabulary containing the unary relation symbol
$\textit{root}$ and a binary relation symbol~$E$.
There exists a class ${\cl A}$ of 
rooted path structures 
of vocabulary $\tau$ 
such that 
$\pemb{\cl A}$ and $p$-$\lgsh$
are interreducible, with respect to 
pl-Turing reductions.
\end{theorem}

\begin{proof}
For $0 \leq k < \ell$,
define $\str P_{k,\ell}$ to be the rooted path structure
(on vocabulary $\tau$) with universe 
$\{ p_1, \ldots, p_{\ell+1} \}$ 
where
$\textit{root}^{\str P_{k,\ell}} = \{ p_1 \}$
and
$E^{\str P_{k,\ell}}$
is the union of
$\{ (p_i, p_{i+1}), (p_{i+1}, p_i) ~|~ 
i \in [\ell] \setminus \{ k+1 \} \}$
with $\{ (p_{k+1}, p_{k+2}) \}$.
This structure has at most one unfoldable edge, namely
the edge $(p_{k+2}, p_{k+3})$ (if it exists);
this edge is unfoldable of degree $1$, but not of degree $2$.
Let ${\cl A}$ be the class containing all such structures~$\str P_{k,\ell}$.
It follows from Theorem~\ref{thm:main} 
that the problem $\pemb{\cl A}$
reduces to $p$-$\lgsh$.

We thus establish that $p$-$\lgsh$ reduces to $\pemb{\cl A}$.
Let $({\str G}, s, t, k, \ell)$ be an instance of 
the problem $p$-$\lgsh$.
The reduction produces the instance $({\str P_{k, \ell}}, {\str G'})$
where ${\str G'}$ is defined as follows:
\begin{eqnarray*}
G' &:=& G \cup \{ q_1, \ldots, q_{\ell-k} \},\\
\textit{root}^{\str G'} &:=& \{ s \},\\
E^{\str G'} &:=& E^{\str G} \cup \{ (t, q_1) \} \cup 
\{ (q_i,q_{i+1}), (q_{i+1}, q_i) ~|~ i \in [\ell-k-1] \}.
\end{eqnarray*}


Suppose that the original instance is a yes instance of
$p$-$\lgsh$.  If ${\str G}$ contains a path of length at least 
$\ell$ with endpoint $s$, then 
the structure ${ \str P_{k,\ell} }$ admits an injective
homomorphism to ${\str G'}$, namely, by simply mapping
the elements of ${ \str P_{k,\ell} }$ onto the path.
If ${\str G}$ contains an $s$-$t$ path of length exactly $k$,
then there is also an injective homomorphism; namely,
the elements $p_1, \ldots, p_{k+1}$ are mapped onto
the $s$-$t$ path, with $p_1$ mapped to $s$ and
$p_{k+1}$ mapped to $t$,
and the elements $p_{k+2}, \ldots, p_{\ell + 1 }$
are mapped to $q_1, \ldots, q_{\ell-k}$, respectively.

Suppose that the created instance is a yes instance of 
$\pemb{\cl A}$; let $h$ be the injective homomorphism witnessing this.
If $q_1$ is not in the image of $h$, then none of the points $q_i$
are, and hence $h$ is an injective homomorphism into ${\str G}$,
implying that ${\str G}$ has a path of length $\ell$ with endpoint $s$.
If $q_1$ is in the image of $h$, then it must hold that an edge
of ${\str P_{k,\ell}}$ maps onto $(t, q_1)$, since any path from
$s$ to $q_1$ in ${\str G'}$ must touch $t$ immediately prior to touching 
$q_1$.  But since $(q_1, t) \notin E^{\str G'}$,
the only edge of ${\str P_{k,\ell}}$ that can map onto 
$(t,q_1)$ is $(p_{k+1}, p_{k+2})$, implying that the image of
$p_1, \ldots, p_{k+1}$ under $h$ yields an $s$-$t$ path of 
length $k$ in ${\str G}$.
\end{proof}

\subparagraph*{Acknowledgements} We thank the anonymous referees for their careful reading of the manuscript.  Following their suggestions we added a number of examples, comments and discussions.
The first author was supported by the
Spanish Project MINECO COMMAS TIN2013-46181-C2-R, Basque Project GIU15/30, and Basque Grant UFI11/45.
The authors wish to thank Eric Allender, Yijia Chen,
and Michael Elberfeld for useful comments and discussion.







\end{document}